\documentclass[11pt,a4paper]{article}

\usepackage{fullpage}
\usepackage[utf8x]{inputenc}
\usepackage{amsmath, amsthm}
\usepackage{wrapfig,graphicx,amssymb,textcomp,array,amsmath}
\usepackage{algpseudocode} 
\usepackage{enumerate}
\usepackage{enumitem}
\usepackage{multirow}
\usepackage{tabularx}
\algtext*{EndWhile}
\algtext*{EndIf}
\usepackage{algorithm}
\usepackage{color}
\usepackage[titletoc,title]{appendix}

\newcommand{\CH}[1]{\text{$CH(#1)$}}
\newcommand{\require}{\textbf{Input: }}
\newcommand{\ensure}{\textbf{Output: }}
\newcommand{\ti}{\frac{3+4\pi}{3}}

\newcommand{\dist}[2]{|#1#2|}
\newcommand{\distG}[3]{|#1#2|_{#3}}
\newcommand{\TC}{C_1}
\newcommand{\BC}{C_2}

\newcommand{\linepq}[2]{\ell(#1,#2)}
\newcommand{\ray}[2]{R({#1}{\rightarrow}{#2})}
\newcommand{\PS}{\textsf{\sc Deg3PlaneSpanner}}
\newcommand{\CS}{\textsf{\sc Matching}}
\newcommand{\lune}[2]{L(#1,#2)}
\newcommand{\disk}[2]{D(#1,#2)}
\newcommand{\pathG}[3]{\delta_{#3}(#1,#2)}
\newcommand{\lattice}{\Lambda}

\usepackage{yhmath}
\newcommand{\rwf}[1]{\wideparen{#1}}

\title{Towards Plane Spanners of Degree 3}

\author{
Ahmad Biniaz\thanks{Carleton University, Canada. Supported by NSERC.}
\and
Prosenjit Bose\footnotemark[2]
\and
Jean-Lou De Carufel\thanks{University of Ottawa, Canada. Supported by NSERC.}
\and
Cyril Gavoille\thanks{University of Bordeaux, France. Supported by the ANR DESCARTES Project.}
\and
Anil Maheshwari\footnotemark[2]
\and
Michiel Smid\footnotemark[2]
}

\date{\today}
\newtheorem{lemma}{Lemma}

\newtheorem{theorem}{Theorem}
\newtheorem{observation}{Observation}
\newtheorem*{problem*}{Problem}

\begin{document}

\maketitle
\begin{abstract}
Let $S$ be a finite set of points in the plane. In this paper we consider the problem of computing plane spanners of degree at most three for $S$.
\begin{enumerate}
 \item If $S$ is in convex position, then we present an algorithm that constructs a plane $\ti$-spanner for $S$ whose vertex degree is at most 3. 
\item If $S$ is the vertex set of a non-uniform rectangular lattice, then we present an algorithm that constructs a plane $3\sqrt{2}$-spanner for $S$ whose vertex degree is at most 3. 
\item If $S$ is in general position, then we show how to compute plane degree-3 spanners for $S$ with a linear number of Steiner points.
\end{enumerate}
\end{abstract}

\section{Introduction}
\label{introduction-section}

Let $S$ be a finite set of points in the plane. A \emph{geometric graph} is a graph $G=(S,E)$ with vertex set $S$ and edge set $E$ consisting of line segments connecting pairs of vertices. The {\em length} (or {\em weight}) of any edge $(p,q)$ in $E$ is defined to be the Euclidean distance $\dist{p}{q}$ between $p$ and $q$. The {\em length} of any path in $G$ is defined to be the sum of the lengths of the edges on this path. For any two vertices $p$ and $q$ of $S$, their {\em shortest-path distance} in $G$, denoted by $\distG{p}{q}{G}$, is a minimum length of any path in $G$ between $p$ 
and $q$. For a real number $t \geqslant 1$, the graph $G$ is a 
$t$-\emph{spanner} of $S$ if for any two points $p$ and $q$ in $S$, $\distG{p}{q}{G}\leq t\dist{p}{q}$. The smallest value of $t$ for which $G$ is a $t$-spanner is called the \emph{stretch factor} of $G$. A large number of algorithms have been proposed for constructing $t$-spanners for any given point set; see the book by Narasimhan and Smid~\cite{ns-gsn-07}. 

The {\em degree} of a spanner is defined to be its maximum vertex degree. Note that 3 is a lower bound on the degree of a $t$-spanner, for every constant $t>1$, because a Hamiltonian path through a set of points arranged in a grid has unbounded stretch factor (see~\cite{ns-gsn-07} for more details). Even for points that are in convex position, 3 is a lower bound on the degree of a spanner (see Kanj \emph{et al.}~\cite{Kanj2016}).\footnote{It can be shown that the stretch factor is $\Omega(\sqrt{n})$.} Salowe~\cite{Salowe1994} proved the existence of spanners of degree 4. Das and Heffernan~\cite{Das1996} showed
the existence of spanners of degree 3.

A \emph{plane} spanner is a spanner whose edges do not cross each other. Chew~\cite{c-tipga-86} 
was the first to prove that plane spanners exist. Chew proved that the $L_1$-Delaunay triangulation of a finite point set has stretch factor at most $\sqrt{10}\approx 3.16$ (observe that lengths in this graph are measured in the Euclidean metric). In the journal version~\cite{c-tapga-89}, Chew proves that the Delaunay triangulation based on a convex distance function defined by an equilateral triangle is a $2$-spanner.
Dobkin \emph{et al.}~\cite{dfs-dgaag-90} proved that the $L_2$-Delaunay triangulation is a $t$-spanner for $t = \frac{\pi (1+\sqrt{5})}{2}\allowbreak\approx\allowbreak 5.08$. Keil and Gutwin~\cite{kg-cgwac-92} improved the upper bound on the stretch factor to $t = \frac{4\pi}{3\sqrt{3}}\allowbreak\approx\allowbreak 2.42$. This was subsequently improved by Cui \emph{et al.}~\cite{ckx-sfdtpcp-11} to $t=2.33$ for the case when the point set is in convex position. Currently, the best result is due to Xia~\cite{x-sfdtl-13}, who proved that $t$ is less than $1.998$. 
For points that are in convex position the current best upper bound on the stretch factor of plane spanners is $1.88$ that was obtained by Amani {\em et al.}~\cite{Amani2016}. Regarding lower bounds, by considering the four vertices of a square, it is obvious that a plane $t$-spanner with $t<\sqrt{2}$ does not exist. Mulzer~\cite{m-mdtrng-04} has shown that every plane spanning graph of the 
vertices of a regular $21$-gon has stretch factor at least $1.41611$. Recently, Dumitrescu and Ghosh~\cite{Dumitrescu2016b} improved the lower bound to $1.4308$ for the vertices of a regular $23$-gon.

The problem of constructing bounded-degree spanners that are plane and have small stretch factor has received considerable
attention (e.g., see~\cite{Bonichon2010, Bonichon2015, Bose2012, Bose2005, Bose2009, Kanj2008, Kanj2016, Li2004}).
Bonichon \emph{et al.}~\cite{Bonichon2015} proved the existence of a degree 4 plane spanner with stretch factor $156.82$. A simpler algorithm by Kanj \emph{et al.}~\cite{Kanj2016} constructs a degree 4 plane spanner with stretch factor 20; for points that are in convex position, this algorithm gives a plane spanner of degree at most 3 with the same stretch factor. Dumitrescu and Ghosh~\cite{Dumitrescu2016a} considered plane spanners for uniform grids. For the infinite uniform square grid, they proved the existence of a plane spanner of degree 3 whose stretch factor is at most $2.607$; the lower bound is $1+\sqrt{2}$.  

In this paper we consider bounded-degree plane spanners. In Section~\ref{convex-section} we present an algorithm that computes a plane $\frac{3+4\pi}{3}\approx 5.189$-spanner of degree 3 for points in convex position. In Section~\ref{grid-section} we consider finite non-uniform rectangular grids; we present an algorithm that computes a degree 3 plane spanner whose stretch factor is at most $3\sqrt{2}\approx 4.25$. In Section~\ref{Steiner-section} we show that any plane $t$-spanner for points in the plane that are in general position can be converted to a plane $(t+\epsilon)$-spanner of degree at most 3 that uses a linear number of Steiner points, where $\epsilon > 0$ is an arbitrary small constant.

\section{Preliminaries}
\label{preliminaries-section}
For any two points $p$ and $q$ in the plane, let $pq$ denote the line segment between $p$ and $q$, let $\linepq{p}{q}$ denote the line passing through $p$ and $q$, let $\ray{p}{q}$ denote the ray emanating from $p$ and passing through $q$, and let $\disk{p}{q}$ denote the closed disk that has $pq$ as a diameter. Moreover, let $L(p,q)$ denote the lune of $p$ and $q$, which is the intersection of the two closed disks of radius $|pq|$ that are centered at $p$ and $q$.

Let $S$ be a finite and non-empty set of points in the plane. We denote by $\CH{S}$ the boundary of the convex hull of $S$. The {\em diameter} of $S$ is the largest distance among the distances between all pairs of points of $S$. Any pair of points whose distance is equal to the diameter is called a {\em diametral pair}. Any point of any diametral pair of $S$ is called a {\em diametral point}. A {\em chain} is a sequence of points together with line segments connecting every pair of consecutive vertices.

\begin{observation}
\label{lune-obs}
Let $S$ be a finite set of at least two points in the plane, and let $\{p,q\}$ be any diametral pair of $S$. Then, the points of $S$ lie in $L(p,q)$.
\end{observation}

The following theorem is a restatement of Theorem~7.11 in~\cite{Benson1966}.

\begin{theorem}[See~\cite{Benson1966}]
\label{Benson-thr}
If $C_1$ and $C_2$ are convex polygonal regions with $C_1\subseteq C_2$, then the length of the boundary of $C_1$ is at most the length of the boundary of $C_2$. 
\end{theorem}

By simple calculations, one can verify the correctness of the following lemma; however, you can find a proof of it in~\cite{Amani2016}.
\begin{lemma}
 \label{triangle-spanner}
Let $a$, $b$, and $c$ be three points in the plane, and let $\beta= \angle abc$. Then, $$\frac{|ab|+|bc|}{|ac|}\leqslant \frac{1}{\sin(\beta/2)}.$$
\end{lemma}

\begin{lemma}
\label{angle-in-lune-lemma}
 Let $a$ and $b$ be two points in the plane. Let $c$ be a point that is on the boundary or in the interior of $\lune{a}{b}$. Then, $\angle acb\geqslant \frac{\pi}{3}$.
\end{lemma}
\begin{proof}
Since $c\in \lune{a}{b}$, we have $\dist{c}{a}\leqslant\dist{a}{b}$ and $\dist{c}{b}\leqslant \dist{a}{b}$. Thus, $ab$ is a largest side of the triangle $\bigtriangleup abc$. This implies that $\angle acb$ is a largest internal angle of $\bigtriangleup abc$. Based on this, and since the sum of the internal angles of $\bigtriangleup abc$ is $\pi$, we conclude that $\angle acb\geqslant \frac{\pi}{3}$. 
\end{proof}

\section{Plane Spanners for Points in Convex Position}
\label{convex-section}
In this section we consider degree-3 plane spanners for points that are in convex position. Let $S$ be a finite set of points in the plane that are in convex position. Consider the two chains that are obtained from $\CH{S}$ by removing any two edges. Let $\tau$ be the larger stretch factor of these two chains; notice that $\tau$ is not necessarily determined by the endpoints of the chain. In Section~\ref{double-chain-section} we present an algorithm that computes a plane $(2\tau+1)$-spanner of degree 3 for $S$. Based on that, in Section~\ref{convex-subsection} we show how to compute a plane $\frac{3+4\pi}{3}$-spanner of degree 3 for $S$. Moreover, we show that if $S$ is centrally symmetric, then there exists a plane $(\pi+1)$-spanner of degree 3 for $S$.

\subsection{Spanner for Convex Double Chains}
\label{double-chain-section}
Let $\TC$ and $\BC$ be two chains of points in the plane that are separated by a straight line. Let $S_1$ and $S_2$ be the sets of vertices of $\TC$ and $\BC$, respectively, and assume that $S_1\cup S_2$ is in convex position. Let $\tau$ be a real number. In this section we show that if the stretch factor of each of $\TC$ and $\BC$ is at most $\tau$, then there exists a plane $(2\tau+1)$-spanner for $S_1\cup S_2$ whose degree is 3.

In order to build such a spanner, we join $\TC$ and $\BC$ by a set of edges that form a matching. Thus, the spanner consists of $\TC$, $\BC$, and a set $E$ of edges such that each edge has one endpoint in $\TC$ and one endpoint in $\BC$. The set $E$ is a matching, i.e., no two edges of $E$ are incident on a same vertex. We show how to compute $E$ recursively. Let $(a,b)$ be the closest pair of vertices between $\TC$ and $\BC$; see Figure~\ref{chains-fig}. Add this closest pair $(a,b)$ to $E$. Then remove $(a,b)$ from $\TC$ and $\BC$, and recurse on the two pairs of chains obtained on each side of $\linepq{a}{b}$. Stop the recursion as soon as one of the chains is empty. Given $\TC$ and $\BC$, the algorithm $\CS$ computes a set $E$.

\begin{algorithm}[htb]  
\caption{\CS$(\TC,\BC)$}          
\label{procA} 
\require{Two linearly separated chains $\TC$ and $\BC$ with the vertices of $\TC\cup\BC$ in convex position.}\\
\ensure{A matching between the points of $\TC$ and the points of $\BC$.}
\begin{algorithmic}[1]
    \If {$\TC=\emptyset$ or $\BC=\emptyset$}
	\State \Return $\emptyset$
    \EndIf
    \State $(a,b)\gets$ a closest pair of vertices between $\TC$ and $\BC$ such that $a\in\TC$ and $b\in \BC$
    \State $\TC',\TC''\gets$ the two chains obtained by removing $a$ from $\TC$
    \State $\BC',\BC''\gets$ the two chains obtained by removing $b$ from $\BC$
    \State \Return $\{ab\}\cup \CS(\TC',\BC')\cup \CS(\TC'',\BC'')$   
\end{algorithmic}
\end{algorithm}

\begin{figure}[tb]
  \centering
\includegraphics[draft=false, width=.65\columnwidth]{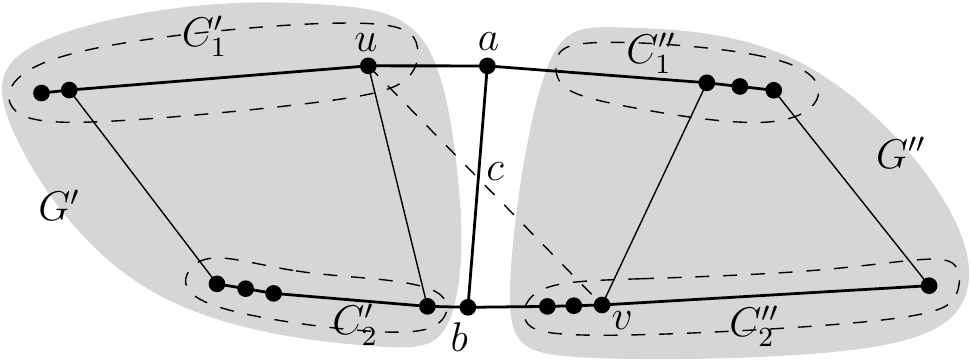}
\caption{Illustration of the proof of Theorem~\ref{chain-spanner-thr}.}
\label{chains-fig}
\end{figure}

In the rest of this section we prove the following theorem.

\begin{theorem}
 \label{chain-spanner-thr}
Let $\TC=(S_1,E_1)$ and $\BC=(S_2,E_2)$ be two linearly separated chains of points in the plane, each with stretch factor at most $\tau$, such that $S_1\cup S_2$ is in convex position. Let $E$ be the set of edges returned by algorithm $\CS(\TC,\BC)$. Then, the graph $G=(S_1\cup S_2, E_1\cup E_2 \cup E)$ is a plane $(2\tau+1)$-spanner for $S_1\cup S_2$ in which the degree of each of the endpoints of $\TC$ and $\BC$ is at most 2 and every other vertex has degree at most 3.
\end{theorem}
\begin{proof}
We prove this theorem by induction on $\min\{|S_1|,|S_2|\}$. As for the base cases, if $|S_1|=0$, then $G=C_2$ is a plane $\tau$-spanner whose vertex degree is at most 2. If $|S_2|=0$, then $G=C_1$ is a plane $\tau$-spanner whose vertex degree is at most 2.

Assume $|S_1|\geqslant 1$ and $|S_2|\geqslant 1$. Let $\ell$ be a line that separates $\TC$ and $\BC$. Without loss of generality assume $\ell$ is horizontal, $\TC$ is above $\ell$, and $\BC$ is below $\ell$. Let $(a,b)$ be the pair of vertices selected by algorithm $\CS$, where $(a,b)$ is a closest pair of vertices between $\TC$ and $\BC$ such that $a\in \TC$ and $b\in \BC$.
Let $\TC'$ and $\TC''$ be the left and right sub-chains of $\TC$, respectively, that are obtained by removing $a$; see Figure~\ref{chains-fig}. We obtain $\BC'$ and $\BC''$ similarly. Note that the chains $\TC'$ and $\BC'$ satisfy the conditions of Theorem~\ref{chain-spanner-thr}. Let $G'$ be the spanner obtained for the vertices of $\TC'$ and $\BC'$. By the induction hypothesis, $G'$ is a plane $(2\tau+1)$-spanner for the vertices of $\TC'\cup \BC'$ in which the degree of each of the endpoints of $\TC'$ and $\BC'$ is at most 2 and every other vertex has degree at most 3. Similarly, let $G''$ be the spanner obtained for the vertices of $\TC''$ and $\BC''$.

Observe that in $G$ the degree of $a$, $b$, the right endpoint of $\TC'$, the right endpoint of $\BC'$, the left endpoint of $\TC''$, and the left endpoint of $\BC''$ is at most 3. Moreover, in $G$, the degree of each endpoint of $\TC$ and $\BC$ is at most 2. Thus, $G$ satisfies the degree condition. As for planarity, since $G'$ and $G''$ are both plane and they are separated by $\linepq{a}{b}$, $G'\cup G''$ is also plane. Moreover, because of convexity, the edges that are incident on each of $a$ and $b$ do not cross any edge of $G'\cup G''$. Therefore, $G$ is plane.

It remains to prove that the stretch factor of $G$ is at most $2\tau+1$. We are going to prove that for any two points $u,v\in S_1\cup S_2$ we have $\distG{u}{v}{G}\leqslant (2\tau+1)\dist{u}{v}$. If both $u$ and $v$ belong to $S_1$, or both belong to $S_2$, then $\distG{u}{v}{G}\leqslant \tau\dist{u}{v}$; this is valid because each of $\TC$ and $\BC$ has stretch factor at most $\tau$. Assume $u\in S_1$ and $v\in S_2$. If $u,v\in G'$ or $u,v\in G''$ then, by the induction hypothesis, $\distG{u}{v}{G}\leqslant (2\tau+1)\dist{u}{v}$.
Thus, it only remains to prove $\distG{u}{v}{G}\leqslant (2\tau+1)\dist{u}{v}$ for the following cases:
(a) $u=a$ and $v\in \BC$,
(b) $u\in \TC$ and $v=b$,
(c) $u\in \TC'$ and $v\in \BC''$, and 
(d) $u\in \TC''$ and $v\in \BC'$.
Because of symmetry we only prove cases (a) and (c). 

First, we prove case (a). Assume $u=a$ and $v\in \BC$. 
Note that \begin{equation}\label{casea}\distG{a}{v}{G}\leqslant \dist{a}{b}+\distG{b}{v}{\BC}\leqslant \dist{a}{v}+\tau\dist{b}{v},\end{equation} where the second inequality is valid since $\dist{a}{b}\leqslant \dist{a}{v}$, by our choice of $(a,b)$, and since $\distG{b}{v}{\BC}\leqslant \tau\dist{b}{v}$, given that the stretch factor of $\BC$ is at most $\tau$. By the triangle inequality we have $\dist{b}{v}\leqslant\dist{a}{b}+\dist{a}{v}$. Since $\dist{a}{b}\leqslant \dist{a}{v}$, we have $\dist{b}{v}\leqslant 2\dist{a}{v}$. By combining this with Inequality~\eqref{casea} we get $\distG{a}{v}{G}\leqslant (2\tau+1)\dist{a}{v},$
which completes the proof for case (a). 

Now, we prove case (c). Assume $u\in \TC'$ and $v\in \BC''$.
Since $S$ is in convex position, the polygon $Q$ formed by $u$, $a$, $v$, and $b$ is convex and its vertices appear in the order $u,a,v,b$.
Note that 
\begin{equation}\label{casec}\distG{u}{v}{G}\leqslant \distG{u}{a}{\TC}+\dist{a}{b}+\distG{b}{v}{\BC}\leqslant \tau\dist{u}{a}+\dist{u}{v}+\tau\dist{b}{v}= \dist{u}{v}+\tau(\dist{u}{a}+\dist{b}{v}),\end{equation}                                                                                                                                                                                                      
where the second inequality is valid since $\dist{a}{b}\leqslant \dist{u}{v}$, by our choice of $(a,b)$, and since $\distG{u}{a}{\TC}\leqslant \tau\dist{u}{a}$ and $\distG{b}{v}{\BC}\leqslant \tau\dist{b}{v}$, given that the stretch factor of each of $\TC$ and $\BC$ is at most $\tau$. Let $c$ be the intersection point of $ab$ and $uv$; see Figure~\ref{chains-fig}. By the triangle inequality, we have $\dist{u}{a}\leqslant \dist{u}{c}+\dist{c}{a}$ and $\dist{b}{v}\leqslant \dist{b}{c}+\dist{c}{v}$. It follows that $\dist{u}{a}+
\dist{b}{v}\leqslant \dist{u}{v}+\dist{a}{b}$. Since $\dist{a}{b}\leqslant \dist{u}{v}$, we have $\dist{u}{a}+
\dist{b}{v}\leqslant 2\dist{u}{v}$. By combining this with Inequality~\eqref{casec} we get $\distG{u}{v}{G}\leqslant (2\tau+1)\dist{u}{v},$ which completes the proof of case (c).
\end{proof}

\subsection{Spanner for Points in Convex Position}
\label{convex-subsection}
In this section we show how to construct plane spanners of degree at most 3 for points that are in convex position.
\begin{theorem}
\label{convex-thr}
Let $S$ be a finite set of points in the plane that is in convex position. Then, there exists a plane spanner for $S$ whose stretch factor is at most $\ti$ and whose vertex degree is at most 3.
\end{theorem}

\noindent The proof of this theorem uses the following result, which will be proved in Subsection~\ref{diametral-lune-section}:

\begin{theorem}
\label{diametral-lune-thr}
 Let $C$ be a convex chain with endpoints $p$ and $q$. If $C$ is in $\lune{p}{q}$, then the stretch factor of $C$ is at most $\frac{2\pi}{3}$.
\end{theorem}

\begin{proof}[Proof of Theorem~\ref{convex-thr}]
The proof is constructive; we present an algorithm that constructs such a spanner for $S$.
The algorithm works as follows. Let $(p,q)$ be a diametral pair of $S$. 
Consider the convex hull of $S$. Let $\TC$ and $\BC$ be the two chains obtained from $\CH{S}$ by removing $p$ and $q$ (and their incident edges). Note that $\TC$ and $\BC$ are separated by $\linepq{p}{q}$. Let $G'$ be the graph on $S\setminus\{p,q\}$ that contains the edges of $\TC$, the edges of $\BC$, and the edges obtained by running algorithm $\CS(\TC,\BC)$. By Theorem~\ref{chain-spanner-thr}, $G'$ is plane and the endpoints of $\TC$ and $\BC$ have degree at most 2. We obtain a desired spanner, $G$, by connecting $p$ and $q$, via their incident edges in $\CH{S}$, to $G'$. This construction is summarized in algorithm $\PS$.

\begin{algorithm}[H]  
\caption{\PS$(S)$}          
\label{procB} 
\require{A non-empty finite set $S$ of points in the plane that is in convex position}\\
\ensure{A plane degree-3 spanner of $S$.}
\begin{algorithmic}[1]
    \State $(p,q)\gets$ a diametral pair of $S$
    \State $\TC,\BC\gets$ the two chains obtained by removing $p$ and $q$ from $\CH{S}$
    \State $E\gets \CH{S}\cup\CS(\TC,\BC)$
    \State \Return $G=(S,E)$   
\end{algorithmic}
\end{algorithm}

Observe that $G$ is plane. Moreover, all vertices of $G$ have degree at most 3; $p$ and $q$ have degree 2. 
Now we show that the stretch factor of $G$ is at most $\ti\approx 5.19$. Note that $G$ consists of $\CH{S}$ and a matching which is returned by algorithm $\CS$. Since $p$ and $q$ are diametral points, then by a result of~\cite{Amani2016}, for any point $s\in S\setminus\{p\}$ we have $\distG{p}{s}{\CH{S}}\leqslant 1.88\dist{p}{s}$. Since $\CH{S}\subseteq G$, we have $\distG{p}{s}{G}\leqslant 1.88\dist{p}{s}$. By symmetry, the same result holds for $q$ and any point $s\in S\setminus\{q\}$. 
Since $(p,q)$ is a diametral pair of $S$, both $\TC$ and $\BC$ are in $\lune{p}{q}$. Based on that, in Theorem~\ref{diametral-lune-thr}, we will see that both $\TC$ and $\BC$ have stretch factor at most $\frac{2\pi}{3}$. Then, by Theorem~\ref{chain-spanner-thr}, the stretch factor of $G'$ is at most $\ti$. Since $G'\subset G$, for any two points $r,s\in S\setminus\{p,q\}$ we have $\distG{r}{s}{G}\leqslant \ti\dist{r}{s}$. Therefore, the stretch factor of $G$ is at most $\ti$. This completes the proof of the theorem.
\end{proof}
A point set $S$ is said to be {\em centrally symmetric} (with respect to the origin), if for every point $p\in S$, the point $-p$ also belongs to $S$.
\begin{theorem}
\label{symmetric-convex-thr}
Let $S$ be a finite centrally symmetric point set in the plane that is in convex position. Then, there exists a plane spanner for $S$ whose stretch factor is at most $\pi+1$ and whose vertex degree is at most 3.
\end{theorem}

\noindent The proof of this theorem uses the following result, which will be proved in Subsection~\ref{diametral-circle-section}:

\begin{theorem}
\label{diametral-circle-thr}
 Let $C$ be a convex chain with endpoints $p$ and $q$. If $C$ is in $\disk{p}{q}$, then the stretch factor of $C$ is at most $\frac{\pi}{2}$.
\end{theorem}
\begin{proof}[Proof of Theorem~\ref{symmetric-convex-thr}]
Let $G$ be the graph obtained by $\PS(S)$. As we have seen in the proof of Theorem~\ref{convex-thr}, $G$ is plane and has vertex degree at most 3. It remains to show that the stretch factor of $G$ is at most $\pi+1$. Let $(p,q)$ be the diametral pair of $S$ that is considered by algorithm $\PS$. Since $S$ is centrally symmetric, all points of $S$ are in $\disk{p}{q}$. Based on that, in Theorem~\ref{diametral-circle-thr}, we will see that both $\TC$ and $\BC$ have stretch factor at most $\frac{\pi}{2}$. Then Theorem~\ref{chain-spanner-thr} implies that the stretch factor of $G$ is at most $\pi+1$. 
\end{proof}

In the following two subsections we will prove Theorems~\ref{diametral-lune-thr} and \ref{diametral-circle-thr}.
Let $C$ be a chain of points. For any two points $u$ and $v$ on $C$ we denote by $\pathG{u}{v}{C}$ the path between $u$ and $v$ on $C$. Recall that $\distG{u}{v}{C}$ denote the length of $\pathG{u}{v}{C}$.

\subsubsection{Proof of Theorem~\ref{diametral-circle-thr}}
\label{diametral-circle-section}
Let $p$ and $q$ be the endpoints of the convex chain $C$. Assume that $C$ is in $\disk{p}{q}$. We are going to show that the stretch factor of $C$ is at most $\frac{\pi}{2}$.

Since $C$ is convex, it is contained in a half-disk of $\disk{p}{q}$, i.e., a half-disk with diameter $pq$. Let $u$ and $v$ be any two points of $C$. We show that $\pathG{u}{v}{C}$ is in $\disk{u}{v}$. Then, by Theorem~\ref{Benson-thr} the length of $\pathG{u}{v}{C}$ is at most the length of the half-arc of $\disk{u}{v}$, which is $\frac{\pi}{2}\dist{u}{v}$.
Without loss of generality assume that $pq$ is horizontal, $p$ is to the left of $q$, and $C$ is above $pq$. Assume that $u$ appears before $v$ while traversing $C$ from $p$ to $q$. See Figure~\ref{diametral-circle-fig}. We consider the following cases. 
\begin{figure}[H]
  \centering
\setlength{\tabcolsep}{0in}
  $\begin{tabular}{ccc}
 \multicolumn{1}{m{.33\columnwidth}}{\centering\includegraphics[draft=false, width=.27\columnwidth]{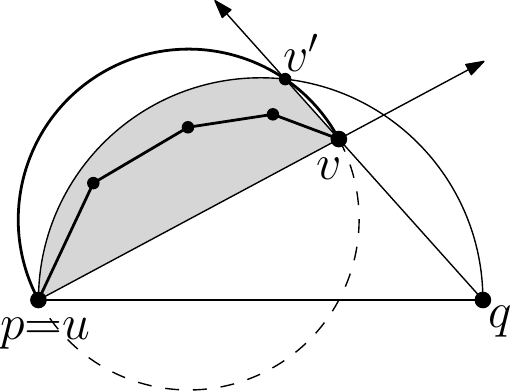}}
&\multicolumn{1}{m{.33\columnwidth}}{\centering\includegraphics[draft=false, width=.27\columnwidth]{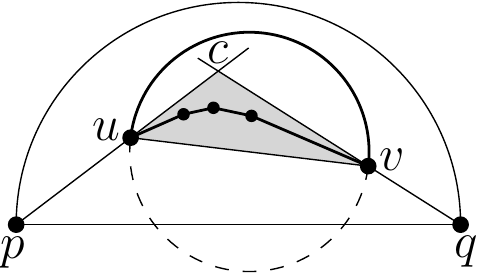}}
&\multicolumn{1}{m{.33\columnwidth}}{\centering\includegraphics[draft=false, width=.27\columnwidth]{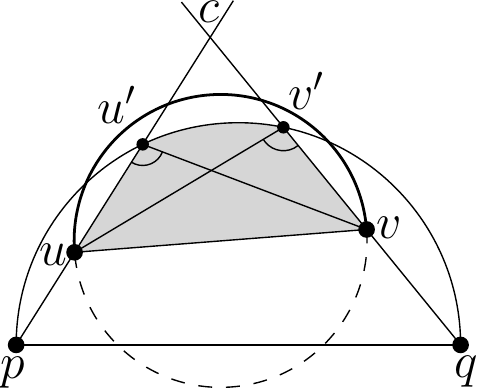}}
\\
(a) & (b) & (c)
\end{tabular}$
  \caption{Proof of Theorem~\ref{diametral-circle-thr}: the path $\pathG{u}{v}{C}$ is inside the shaded regions, where (a) $u=p$ and $v\neq q$, (b) $u\neq p$, $v\neq q$, and $c\in\disk{p}{q}$, and (c) $u\neq p$, $v\neq q$, and $c\notin\disk{p}{q}$.}
\label{diametral-circle-fig}
\end{figure}
\begin{itemize}
 \item $u=p$ and $v=q$. Then $\pathG{p}{q}{C}=C$ is in $\disk{p}{q}$ by the hypothesis.

 \item $u=p$ and $v\neq q$. Let $v'$ be the intersection point of $\ray{q}{v}$ with the boundary of $\disk{p}{q}$. See Figure~\ref{diametral-circle-fig}(a). Observe that $\angle pv'v=\angle pv'q=\frac{\pi}{2}$. Thus, $v'$ is on the boundary of $\disk{p}{v}$. Since two circles can intersect in at most two points, $p$ and $v'$ are the only intersection points of the boundaries of $\disk{p}{q}$ and $\disk{p}{v}$. Thus, the clockwise arc $\rwf{pv'}$ on the boundary of $\disk{p}{q}$ is inside $\disk{p}{v}$. On the other hand, because of convexity, no point of $\pathG{p}{v}{C}$ is to the right of $\ray{p}{v}$ or $\ray{q}{v}$. This implies that $\pathG{p}{v}{C}$ is in $\disk{p}{v}$.
 
 \item $u\neq p$ and $v=q$. The proof of this case is similar to the proof of the previous case.

 \item $u\neq p$ and $v\neq q$. Let $c$ be the intersection point of $\ray{p}{u}$ and $\ray{q}{v}$. Because of convexity, $\pathG{u}{v}{C}$ is in the triangle $\bigtriangleup ucv$. We look at two cases:

\begin{itemize}
 \item $c$ is inside $\disk{p}{q}$. In this case $\angle ucv\geqslant \frac{\pi}{2}$. See Figure~\ref{diametral-circle-fig}(b). This implies that the point $c$, and consequently the triangle $\bigtriangleup ucv$, are inside $\disk{u}{v}$. Thus, $\pathG{u}{v}{C}$ is inside $\disk{u}{v}$. 
 \item $c$ is outside $\disk{p}{q}$. Let $u'$ (resp. $v'$) be the intersection point of $\ray{p}{u}$ (resp. $\ray{q}{v}$) with $\disk{p}{q}$. Because of the convexity and the fact that the path $\pathG{u}{v}{C}$ is contained in $\disk{p}{q}$, this path is inside the region that is bounded by line segments $(u',u)$, $(u,v)$, $(v,v')$ and the boundary of  $\disk{p}{q}$; see the shaded region of Figure~\ref{diametral-circle-fig}(c). Observe that by convexity $\angle uv'v>\angle pv'q=\frac{\pi}{2}$, and $\angle uu'v>\angle pu'q=\frac{\pi}{2}$. Thus, both $u'$ and $v'$ are inside $\disk{u}{v}$. Consequently, the clockwise arc $\rwf{u'v'}$ on the boundary of $\disk{p}{q}$ is inside $\disk{u}{v}$. Therefore, $\pathG{u}{v}{C}$ is inside $\disk{u}{v}$.
\end{itemize}
\end{itemize}

\subsubsection{Proof of Theorem~\ref{diametral-lune-thr}}
\label{diametral-lune-section}
Let $p$ and $q$ be the endpoints of the convex chain $C$. Assume that $C$ is in $\lune{p}{q}$. We are going to show that the stretch factor of $C$ is at most $\frac{2\pi}{3}$.
 
Since $C$ is convex, it is contained in a half-lune of $\lune{p}{q}$; i.e., a portion of $\lune{p}{q}$ that is obtained by cutting it through $pq$. Let $u$ and $v$ be any two points of $C$. We show that the length of $\pathG{u}{v}{C}$ is at most $\frac{2\pi}{3}$ times $\dist{u}{v}$. 
Without loss of generality assume $pq$ is horizontal, $p$ is to the left of $q$, and $C$ is above $pq$. Assume that $u$ appears before $v$ when traversing $C$ from $p$ to $q$. See Figure~\ref{diametral-lune-fig}. 
We consider two cases: (1) $u=p$ or $v=q$, (2) $u\neq p$ and $v\neq q$.

\begin{enumerate}
 \item $u=p$ or $v=q$. Without loss of generality assume $u=p$. If $v=q$ then $\pathG{p}{q}{C}$ is inside the half-lune of $\lune{p}{q}$ and by Theorem~\ref{Benson-thr} the length of $\pathG{p}{q}{C}$ is at most the length of the half-lune, which is $\frac{2\pi}{3} \dist{p}{q}$. Assume $v\neq q$. For simplicity, we assume that $\dist{p}{q}=1$. Let $\alpha=\angle pqv$ and $x=\dist{q}{v}$. If $\alpha=0$, then the length of $\pathG{p}{v}{C}$ is equal to $|pv|$ and we are done. Thus, assume that $\alpha>0$. Moreover, since $v\neq q$, $x>0$. We consider the following two cases: (a) $\alpha\leqslant \frac{\pi}{3}$, (b) $\alpha> \frac{\pi}{3}$.
 \begin{enumerate}
    \item $\alpha\leqslant \frac{\pi}{3}$. Let $v'$ be the intersection point of $\ray{q}{v}$ and $\lune{p}{q}$. Because of the convexity and the fact that the path $\pathG{p}{v}{C}$ is contained in $\lune{p}{q}$, this path is inside the region that is bounded by line segments $(p,v)$, $(v,v')$ and the boundary of  $\lune{p}{q}$; see the shaded region of Figure~\ref{diametral-lune-fig}(a). By Theorem~\ref{Benson-thr} we have the following inequality: $$\distG{p}{v}{C}\leqslant |vv'|+ |\rwf{pv'}|=1-x+\alpha,$$
    where $|\rwf{pv'}|$ is the length of the clockwise arc from $p$ to $v'$ which has radius 1 and is centered at $q$. Note that $$\dist{p}{v}=\sqrt{(x \sin\alpha)^2+(1-x \cos\alpha)^2}.$$
    Define 
    \begin{equation}
      f(x,\alpha)=\frac{1-x+\alpha}{\sqrt{(x \sin\alpha)^2+(1-x \cos\alpha)^2}}.
     \end{equation}
    Then $\frac{\distG{p}{v}{C}}{\dist{p}{v}}\leqslant f(x,\alpha)$. In Appendix~\ref{appA}, we will show that $f(x,\alpha)\leqslant 2.04738\leqslant \frac{2\pi}{3}$ for $0< x\leqslant 1$ and $0< \alpha\leqslant \frac{\pi}{3}$. 

\begin{figure}[H]
  \centering
\setlength{\tabcolsep}{0in}
  $\begin{tabular}{cc}
 \multicolumn{1}{m{.5\columnwidth}}{\centering\includegraphics[draft=false, width=.28\columnwidth]{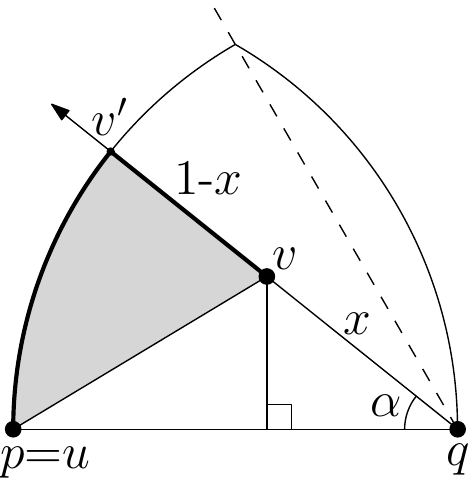}}
&\multicolumn{1}{m{.5\columnwidth}}{\centering\includegraphics[draft=false, width=.28\columnwidth]{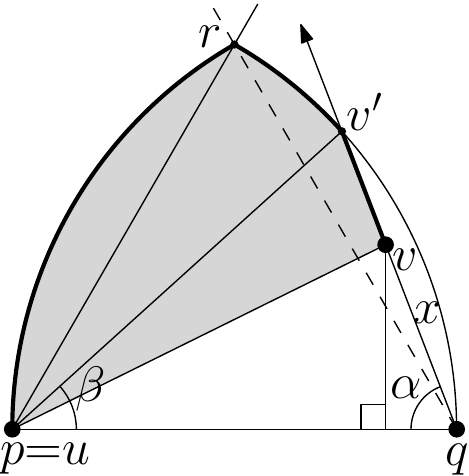}}
\\
(a) & (b) 
\\\multicolumn{1}{m{.5\columnwidth}}{\centering\includegraphics[draft=false, width=.28\columnwidth]{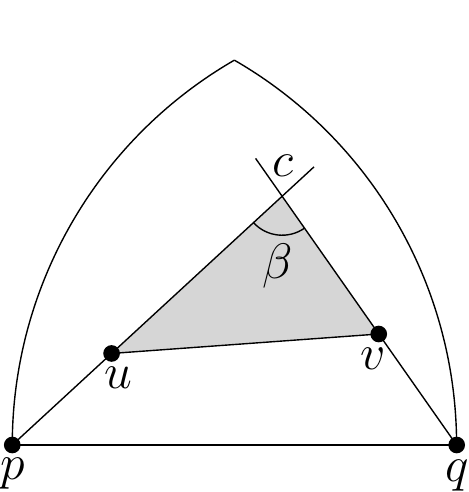}}
&\multicolumn{1}{m{.5\columnwidth}}{\centering\includegraphics[draft=false, width=.28\columnwidth]{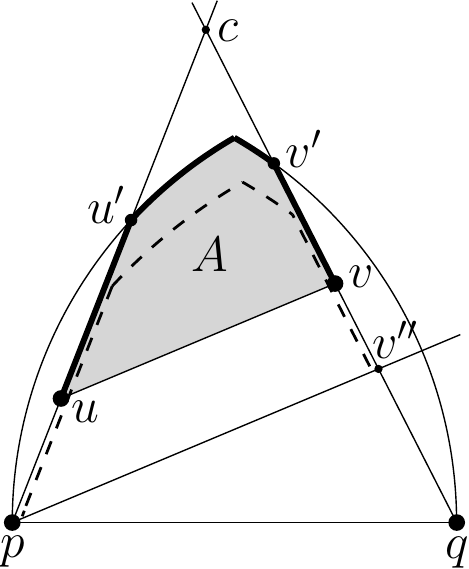}}
\\ (c)&(d)
\end{tabular}$
  \caption{Proof of Theorem~\ref{diametral-lune-thr}: the path $\pathG{u}{v}{C}$ is inside the shaded regions, where (a) $u=p$ and $\alpha\leqslant \frac{\pi}{3}$, (b) $u= p$ and $\alpha> \frac{\pi}{3}$, (c) $u\neq p$, $v\neq q$, and $c\in\lune{p}{q}$, and (d) $u\neq p$, $v\neq q$, and $c\notin\lune{p}{q}$.}
\label{diametral-lune-fig}
\end{figure}
 
    \item $\alpha> \frac{\pi}{3}$. Let $v'$ be the intersection point of $\ray{q}{v}$ and $\lune{p}{q}$, and let $r$ be the topmost intersection point of the two circles of radius $|pq|$ that are centered at $p$ and $q$, i.e., the highest point of $\lune{p}{q}$. Because of the convexity and the fact that the path $\pathG{p}{v}{C}$ is contained in $\lune{p}{q}$, this path is inside the region that is bounded by line segments $(p,v)$, $(v,v')$ and the boundary of  $\lune{p}{q}$; see the shaded region of Figure~\ref{diametral-lune-fig}(b). By Theorem~\ref{Benson-thr} we have the following inequality: $$\distG{p}{v}{C}\leqslant |vv'|+ |\rwf{rv'}|+|\rwf{pr}|,$$
    where $|\rwf{rv'}|$ is the length of the clockwise arc from $r$ to $v'$ which has radius 1 and is centered at $p$ and $|\rwf{pr}|$ is the length of the clockwise arc from $p$ to $r$ which has radius 1 and is centered at $q$. Let $\beta=\angle v'pq$. Then, $|\rwf{rv'}|=\frac{\pi}{3}- \beta$. Note that the triangle $\bigtriangleup pqv'$ is isosceles. Thus, $\beta = \pi- 2\alpha$ and $\dist{q}{v'}=2\cos\alpha$. Thus $\dist{v}{v'}=2\cos\alpha-x$. Therefore, $$\distG{p}{v}{C}\leqslant (2\cos\alpha-x)+ \left(\frac{\pi}{3}-\beta\right)+\frac{\pi}{3}=2(\alpha+\cos\alpha)-\left(x+\frac{\pi}{3}\right).$$ Define 
    \begin{equation}
      g(x,\alpha)=\frac{2(\alpha+\cos\alpha)-(x+\pi/3)}{\sqrt{(x \sin\alpha)^2+(1-x \cos\alpha)^2}}.
    \end{equation}
   Then $\frac{\distG{p}{v}{C}}{\dist{p}{v}}\leqslant g(x,\alpha)$. In Appendix~\ref{appB}, we will show that $g(x,\alpha)\leqslant\frac{2\pi}{3}$ for $0< x\leqslant 2\cos\alpha$ and $\frac{\pi}{3}\leqslant \alpha\leqslant \frac{\pi}{2}$.
\end{enumerate}

 \item $u\neq p$ and $v\neq q$. Let $c$ be the intersection point of rays $\ray{p}{u}$ and $\ray{q}{v}$. We differentiate between two cases: (a) $c$ is inside $\lune{p}{q}$, (b) $c$ is outside $\lune{p}{q}$.
\begin{enumerate}
\item $c$ is inside $\lune{p}{q}$. Because of the convexity, $\pathG{u}{v}{C}$ is inside the triangle $\bigtriangleup uvc$; see Figure~\ref{diametral-lune-fig}(c). Thus $\distG{u}{v}{C}\leqslant \dist{u}{c}+\dist{v}{c}$. Let $\beta=\angle ucv$. By Lemma~\ref{angle-in-lune-lemma} we have $\beta\geqslant \frac{\pi}{3}$. Based on this, and by Lemma~\ref{triangle-spanner}, we have $$\frac{\distG{u}{v}{C}}{\dist{u}{v}}\leqslant \frac{\dist{u}{c} +\dist{v}{c}}{\dist{u}{v}}\leqslant \frac{1}{\sin(\beta/2)}\leqslant \frac{1}{\sin(\pi/6)}=2<\frac{2\pi}{3}.$$
\item $c$ is outside $\lune{p}{q}$. We reduce this case to case 1 where $u=p$ or $v=q$. See Figure~\ref{diametral-lune-fig}(d). Let $u'$ and $v'$ be the intersection points of $\ray{p}{u}$ and $\ray{q}{v}$ with $\lune{p}{q}$, respectively. Let $\delta_{uv}$ be the convex chain consisting of $uu'$, $vv'$, and the portion of the boundary of $\lune{p}{q}$ that is between $u'$ and $v'$ ($\delta_{uv}$ is the bold chain in Figure~\ref{diametral-lune-fig}(d)). Observe that the length of $\pathG{u}{v}{C}$ is at most the length of $\delta_{uv}$. Thus $\frac{\distG{u}{v}{C}}{\dist{u}{v}}\leqslant\frac{|\delta_{uv}|}{\dist{u}{v}}$. Let $A$ be the convex region that is bounded by $\delta_{uv}$ and the segment $(u,v)$. Let $A'$ be obtained from $A$ by a homothetic transformation with respect to the center $c$ and scale factor $\min\{\frac{|cp|}{|cu|},\frac{|cq|}{|cv|}\}$. Assume $\frac{|cp|}{|cu|}\leqslant\frac{|cq|}{|cv|}$. Then, after this transformation $u$ lies on $p$ as shown in Figure~\ref{diametral-lune-fig}(d). Let $v''$ be the point on $qv$ where $v$ ended up after this transformation. Let $\delta_{pv''}$ (the dashed chain in Figure~\ref{diametral-lune-fig}(d)) be the chain obtained from $\delta_{uv}$ after this transformation. Since this transformation preserves ratios of distances we have $\frac{|\delta_{pv''}|}{\dist{p}{v''}}=\frac{|\delta_{uv}|}{\dist{u}{v}}$. Thus, $\frac{\distG{u}{v}{C}}{\dist{u}{v}}\leqslant \frac{|\delta_{pv''}|}{\dist{p}{v''}}$. Note that $\delta_{pv''}$ is in $\lune{p}{q}$. To obtain an upper bound on $\frac{|\delta_{pv''}|}{\dist{p}{v''}}$, we apply case 1 where $u=p$ and $v''$ plays the role of $v$. 
\end{enumerate}
\end{enumerate}

\section{Non-Uniform Rectangular Grid}
\label{grid-section}
In this section we build a plane spanner of degree three for the point set of the vertices of a non-uniform rectangular grid. In a finite non-uniform $m\times k$ grid, $\lattice$, the vertices are arranged on the intersections of $m$ horizontal and $k$ vertical lines. The distances between the horizontal lines and the distances between the vertical lines are chosen arbitrary. The total number of vertices of $\lattice$\textemdash the number of points of the underlying point set\textemdash is $n=m\cdot k$. 

\begin{figure}[H]
  \centering
\setlength{\tabcolsep}{0in}
  $\begin{tabular}{cc}
 \multicolumn{1}{m{.55\columnwidth}}{\centering\includegraphics[draft=false, width=.54\columnwidth]{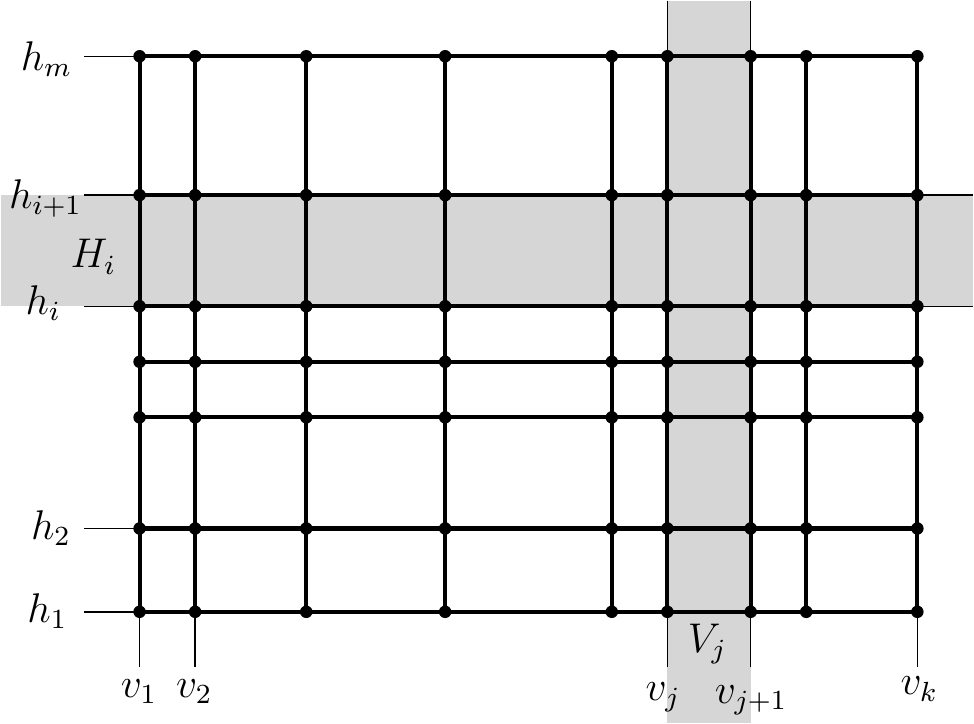}}
&\multicolumn{1}{m{.45\columnwidth}}{\centering\includegraphics[draft=false, width=.44\columnwidth]{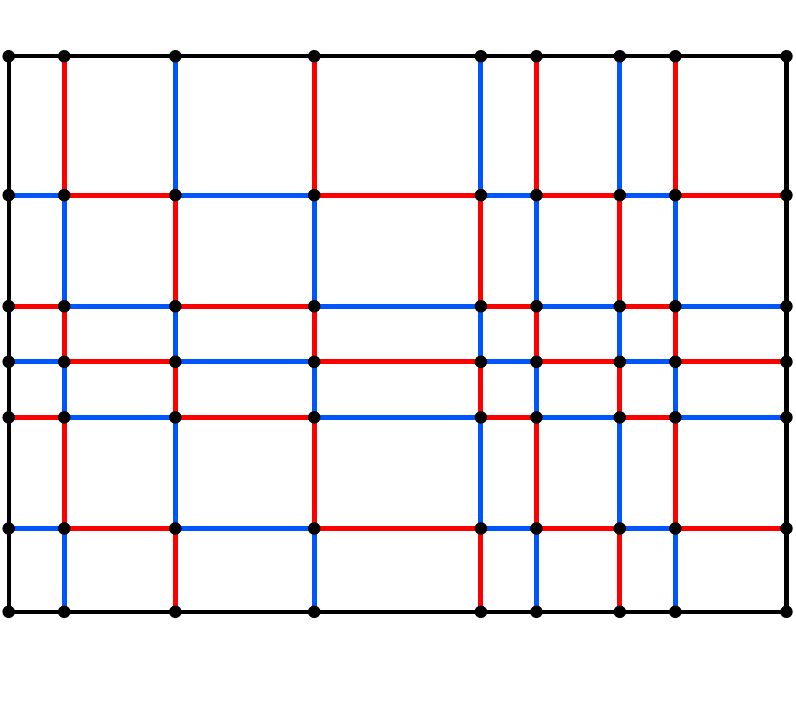}}
\\
(a) & (b)
\end{tabular}$
\caption{The grid $\lattice$: (a) horizontal and verticals slabs, and (b) red and blue staircases.}
\label{augmented-lattice-fig}
\end{figure}

If $m\in\{1,2\}$ or $k\in\{1,2\}$ then $\lattice$ is a plane spanner whose degree is at most 3 and whose stretch factor is at most $\sqrt{2}$. Assume $m\geqslant 3$ and $k\geqslant 3$. We present an algorithm that constructs a degree-3 plane spanner, $G$, for the points of $\lattice$. 

Let $h_1, \dots, h_m$ be the horizontal lines of $\lattice$ from bottom to top. Similarly, let $v_1,\dots, v_k$ be the vertical lines of $\lattice$ from left to right. For each pair $(i,j)$ where $1\leqslant i\leqslant m$ and $1\leqslant j\leqslant k$, we denote by $p_{i,j}$ the vertex of $\lattice$ that is the intersection point of $h_i$ and $v_j$. The vertices that are on $h_1$, $h_m$, $v_1$, or $v_k$ are referred to as {\em boundary vertices}; other vertices are referred to as {\em internal vertices}. The edges both of whose endpoints are boundary vertices are referred to as {\em boundary edges}; other edges are referred to as {\em internal edges}. The grid $\lattice$ consists of $m-1$ horizontal slabs (rows) and $k-1$ vertical slabs (columns). Each horizontal slab $H_i$, with $1\leqslant i<m$, is bounded by consecutive horizontal lines $h_i$ and $h_{i+1}$. Each vertical slab $V_j$, with $1\leqslant j<k$, is bounded by consecutive vertical lines $v_j$ and $v_{j+1}$. See Figure~\ref{augmented-lattice-fig}(a). For each slab we define the {\em width} of that slab as the distance between the two parallel lines on its boundary. 

We partition the set of internal edges of $\lattice$ into two sets: a set $R$ of red edges and a set $B$ of blue edges. See Figure~\ref{augmented-lattice-fig}(b). The set $R$ is the union of the following edge sets:
\begin{align}
     \notag &\{(p_{i,j},p_{i+1,j}): 2\leqslant i\leqslant m-1, 2\leqslant j\leqslant k-1 \text{, $i$ and $j$ are even} \},\\
   \notag &\{(p_{i,j},p_{i+1,j}): 2\leqslant i\leqslant m-1, 2\leqslant j\leqslant k-1 \text{, $i$ and $j$ are odd} \},\\
   \notag &\{(p_{i,j},p_{i,j+1}): 2\leqslant i\leqslant m-1, 2\leqslant j\leqslant k-1 \text{, $i$ and $j$ are even} \},\\
  \notag & \{(p_{i,j},p_{i,j+1}): 2\leqslant i\leqslant m-1, 2\leqslant j\leqslant k-1 \text{, $i$ and $j$ are odd} \}.
\end{align}

The set $B$ consists of all other internal edges. As shown in Figure~\ref{augmented-lattice-fig}(b), the edges of $R$ (and also the edges of $B$) form staircases in $\lattice$; a {\em staircase} is a maximal directed path in $\lattice$ consisting of an alternating sequence of horizontal and vertical internal edges, in which each horizontal edge is traversed from left to right and each vertical edge is traversed from top to bottom. Each internal vertex is incident on two red edges and two blue edges. Moreover, each of the staircases formed by red edges is next to either two staircases of blue edges or a stair case of blue edges and the boundary.

We know that $\lattice$ is a plane $\sqrt{2}$-spanner of degree 4. We present an algorithm that constructs, from $\lattice$, a plane graph whose degree is 3 and whose stretch factor is $3\sqrt{2}$. Let $G'$ be the graph obtained by removing all red edges from $\lattice$ as in Figure~\ref{G-G-prime-fig}(a). Since in $\lattice$ every internal vertex is incident on two red edges, in $G'$ these vertices have degree 2. Initialize $E'$ to be the empty set. The algorithm iterates over all slabs, $\{H_1, \dots,\allowbreak   H_{m-1},\allowbreak   V_1,\allowbreak  \dots,\allowbreak   V_{k-1}\}$, in a non-decreasing order of their widths. Let $S$ be the current slab. The algorithm considers the red edges in $S$ from left to right if $S$ is horizontal and bottom-up if $S$ is vertical (however, this ordering does not
matter). Let $e=(a,b)$ be the current red edge. The algorithm adds $e$ to $E'$ if both endpoints of $e$ have degree 2 in $G'\cup E'$, i.e., $\deg_{G'}(a)+\deg_{E'}(a)=2$ and $\deg_{G'}(b)+\deg_{E'}(b)=2$. See Figure~\ref{G-G-prime-fig}(b). Note that $E'$ and the edge set of $G'$ are disjoint. At the end of this iteration, let $G$ be the graph obtained by taking the union of $G'$ and $E'$. We will show that $G$ is a plane $3\sqrt{2}$-spanner of degree three for the point set of the vertices of $\lattice$.

\begin{figure}[H]
  \centering
\setlength{\tabcolsep}{0in}
  $\begin{tabular}{cc}
 \multicolumn{1}{m{.5\columnwidth}}{\centering\includegraphics[draft=false, width=.47\columnwidth]{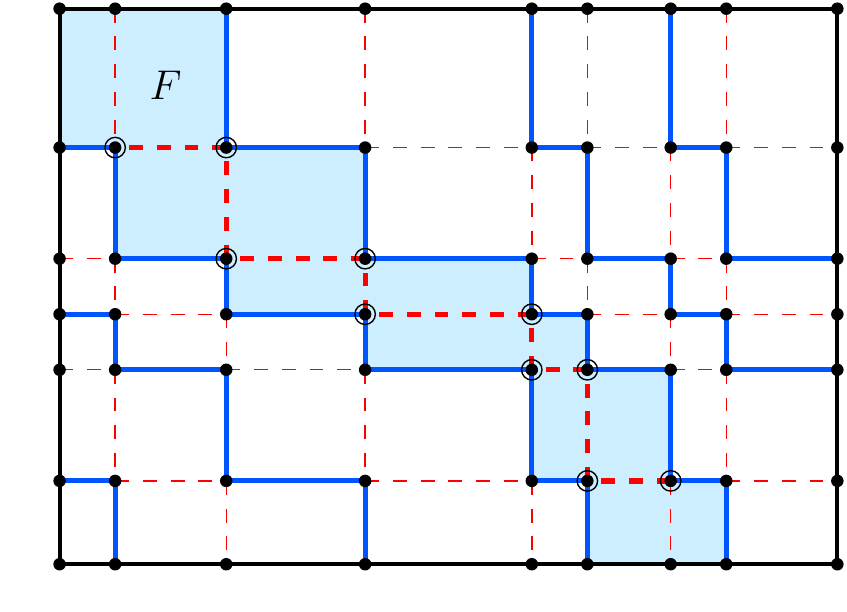}}
&\multicolumn{1}{m{.5\columnwidth}}{\centering\includegraphics[draft=false, width=.47\columnwidth]{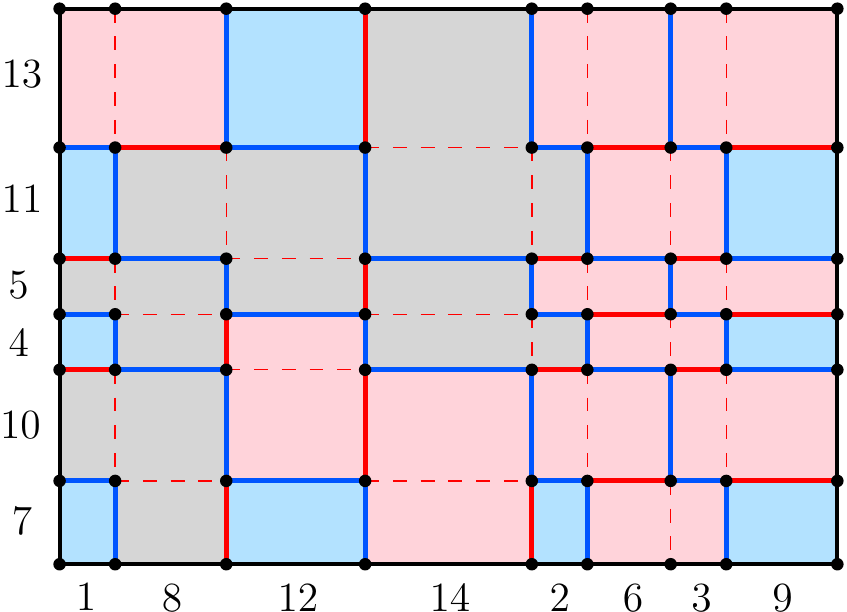}}
\\
(a) & (b)
\end{tabular}$
  \caption{(a) The graph $G'$ that is obtained by removing the red edges (the reflex vertices of the face $F$ are marked, and the path between them is highlighted), and (b) the graph $G$ that is the union of $G'$ and $E'$; the bold red edges belong to $E'$. The numbers close to the slabs show the order in which the slabs are considered.}
\label{G-G-prime-fig}
\end{figure}

The graph $G$ is plane because it is a subgraph of $\lattice$. As for the degree constraint, since we add edges only between vertices that have degree 2 in $G'\cup E'$ (at most one edge per vertex), no vertex of degree 4 can appear. Thus, $G$ has maximum degree 3. It only remains to show that $G$ is a $3\sqrt{2}$-spanner. Before that, we review some properties of $G$.

A {\em cell} of $\lattice$ refers to a region that is bounded by two consecutive horizontal lines and two consecutive vertical lines of $\lattice$. Every face of $G'$ is the union of the cells of $\lattice$ that are between two consecutive blue staircases and the boundary of $\lattice$; see Figure~\ref{G-G-prime-fig}(a). Consider one iteration of the algorithm. We refer to the red edges that are not in $G'\cup E'$ as {\em missing edges}. For every face $F$ in $G'\cup E'$ we define the set of {\em reflex vertices} of $F$ as the vertices on the boundary of $F$ that are incident on two missing edges in $F$; reflex vertices have degree 2 in $G'\cup E'$. Notice that there exists a path of missing edges that connects the reflex vertices of $F$; see Figure~\ref{G-G-prime-fig}(a). Since the algorithm connects the reflex vertices by missing edges, all faces that have more than one reflex vertex have been broken into subfaces. Thus at the end of the algorithm every face of $G$ contains at most one reflex vertex. Therefore, every face of $G$ consists of one cell, two cells, or three cells. We refer to these faces as {\em 1-cell}, {\em 2-cell}, and {\em 3-cell} faces, respectively. See the shaded faces in Figure~\ref{G-G-prime-fig}$($b$)$.

\begin{observation}
\label{G-obs}
Each face in $G$ is either a 1-cell face, a 2-cell face, or a 3-cell face.
\end{observation}

\begin{lemma}
\label{missing-border}
Every missing edge that has an endpoint on the boundary of $\lattice$ is in a 2-cell face of $G$.
\end{lemma}
\begin{proof}
The proof is by contradiction. Assume $(a,b)$ is a missing edge that has an endpoint on the boundary of $\lattice$ and lies in a 3-cell face $F$ of $G$. Without loss of generality assume $a$ is on the boundary. Thus, in $G$, $b$ is the reflex vertex of $F$ and has degree 2. At the moment the algorithm considers $(a,b)$, both $a$ and $b$ have degree 2, and thus, adds $(a,b)$ to $E'$. Since all edges of $E'$, including $(a,b)$, are also edges of $G$, we get a contradiction with $(a,b)$ being a missing edge.
\end{proof}

\begin{figure}[htb]
  \centering
\setlength{\tabcolsep}{0in}
  $\begin{tabular}{ccc}
 \multicolumn{1}{m{.33\columnwidth}}{\centering\includegraphics[draft=false, width=.19\columnwidth]{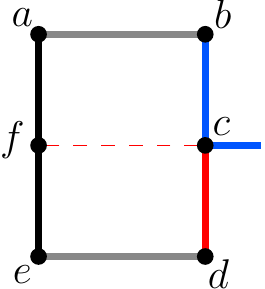}}
&\multicolumn{1}{m{.33\columnwidth}}{\centering\includegraphics[draft=false, width=.2\columnwidth]{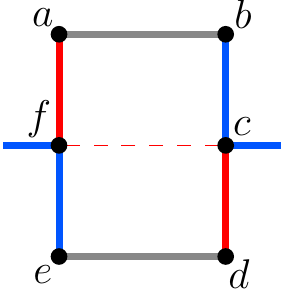}}
 &\multicolumn{1}{m{.33\columnwidth}}{\centering\includegraphics[draft=false, width=.22\columnwidth]{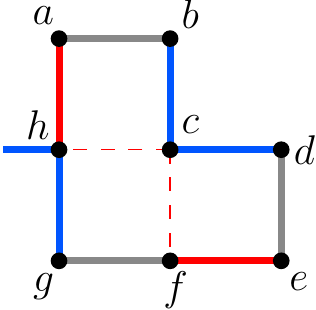}}
\\
(a) & (b)&(c)
\end{tabular}$
  \caption{Illustration of Lemmas~\ref{2-cell-lemma} and \ref{3-cell-lemma}: (a) $(c,f)$ is the missing edge of a 2-cell face and $f$ is a boundary vertex, (b) $(c,f)$ is the missing edge of a 2-cell face and none of $c$ and $f$ are boundary vertices, and (c) $(c,h)$ is a missing edge of a 3-cell face.}
\label{cell-fig}
\end{figure}

\begin{lemma}
\label{2-cell-lemma}
Let $(c,f)$ be the missing edge of a 2-cell face in $G$. Then, $\distG{c}{f}{G}\leqslant 3\dist{c}{f}$. 
\end{lemma}
\begin{proof}
Let $F=(a,b,c,d,e,f)$ be a 2-cell face of $G$ with the edge $(c,f)$ is missing. Without loss of generality assume that $(c,f)$ is horizontal, $f$ is to the left of $c$, and $a,b,c,d,e,f$ is the clockwise order of the vertices along the boundary of $F$; see Figures~\ref{cell-fig}(a) and~\ref{cell-fig}(b). Note that $\dist{c}{f}=\dist{a}{b}=\dist{d}{e}$, $\dist{a}{f}=\dist{b}{c}$, and $\dist{e}{f}=\dist{c}{d}$. 
We consider two cases:
\begin{itemize}
 \item One of $c$ and $f$ is a boundary vertex. Without loss of generality assume that $f$ is a boundary vertex, and thus, it has degree 2. See Figure~\ref{cell-fig}(a). Then, $(c,d)$ is a red edge that has been added during the algorithm. At the moment the algorithm considers $(c,f)$, the vertex $c$ has degree 3, because otherwise, the algorithm adds $(c,f)$ to $E'$, and hence to $G$. Thus, the edge $(c,d)$ has been considered before $(c,f)$, which implies that $|cd|\leqslant|cf|$. Thus, the length of the path $(c,d,e,f)$, i.e., $|cd|+|de|+|ef|$, is at most $3|cf|$.

 \item Both $c$ and $f$ are internal vertices. See Figure~\ref{cell-fig}(b). Then, $(a,f)$ and $(c,d)$ are red edges that have been added during the algorithm. At the moment the algorithm considers $(c,f)$, at least one of $c$ and $f$ has degree 3, because otherwise, the algorithm adds $(c,f)$ to $E'$, and hence to $G$. Without loss of generality assume that $c$ has degree 3. Thus, the edge $(c,d)$ has been considered before $(c,f)$, which implies that $|cd|\leqslant|cf|$. Thus, the length of the path $(c,d,e,f)$, i.e., $|cd|+|de|+|ef|$, is at most $3|cf|$.
\end{itemize}

\end{proof}

\begin{lemma}
\label{3-cell-lemma}
Let $(c,h)$ be a missing edge of a 3-cell face in $G$. Then, $\distG{c}{h}{G}\leqslant 3\dist{c}{h}$.
\end{lemma}
\begin{proof}
Let $F=(a,b,c,d,e,f,g,h)$ be a 3-cell face in $G$ with the edge $(c,h)$ missing; see Figure~\ref{cell-fig}(c). Let $a,b,c,d,e,f,g,h$ be the clockwise order of the vertices along the boundary of $F$. Without loss of generality assume that $c$ is the reflex vertex of $F$, the edge $(c,h)$ is horizontal, and $h$ is to the left of $c$. Note that $\dist{c}{h}=\dist{a}{b}$ and $\dist{b}{c}=\dist{a}{h}$. Since $F$ is a 3-cell face, by Lemma~\ref{missing-border}, $f$ is not a boundary vertex. Thus, $(a,h)$ is a red edge that is added during the algorithm. At the moment the algorithm considers $(c,h)$, the vertex $h$ has degree 3, because otherwise, the algorithm adds $(c,h)$ to $E'$. Thus, $(a,h)$ has been considered before $(c,h)$, which implies that $|ah|\leqslant|ch|$. Thus, the length of the path $(c,b,a,h)$, i.e., $|cb|+|ba|+|ah|$, is at most $3|ch|$. 
\end{proof}

\begin{figure}[H]
  \centering
\includegraphics[draft=false, width=.5\columnwidth]{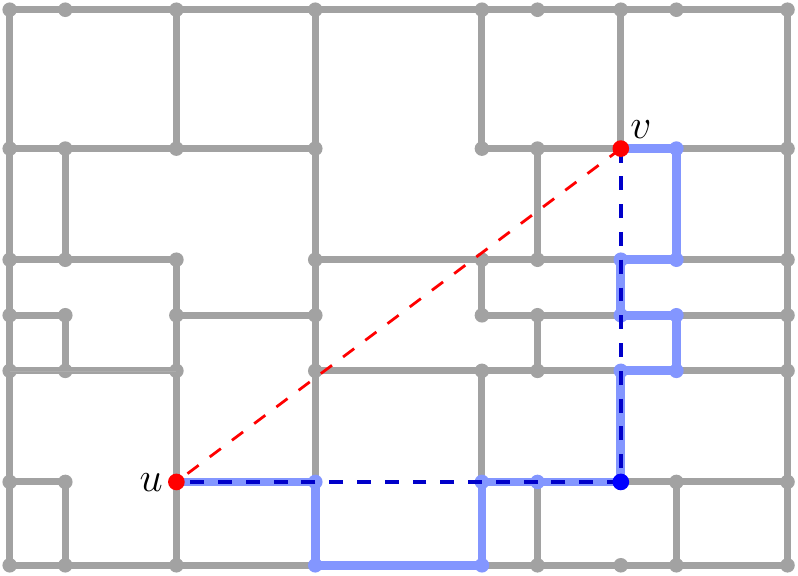}
  \caption{The path $\delta_{uv}$ is shown in dashed blue and the path $\delta'_{uv}$ is shown in bold blue.}
\label{lattice-stretch-fig}
\end{figure}

\begin{theorem}
 \label{lattice-thr}
Let $\lattice$ be a finite non-uniform rectangular grid. Then, there exists a plane spanner for the point set of the vertices of $\lattice$ such that its degree is at most 3 and its stretch factor is at most $3\sqrt{2}$.
\end{theorem}

\begin{proof}
If $\lattice$ has less than three rows or less than three columns, then it is a plane $\sqrt{2}$-spanner of degree 3.
Assume $\lattice$ has at least three rows and at least three columns. Let $G$ be the graph obtained by the algorithm described in this section. Then, $G$ is plane and its vertex degree is at most 3. Since $\lattice$ is a $\sqrt{2}$-spanner, between any two vertices $u$ and $v$, there exists a path $\delta_{uv}$ in $\lattice$ such that $|\delta_{uv}|\leqslant \sqrt{2}|uv|$. By Lemmas~\ref{2-cell-lemma} and \ref{3-cell-lemma}, for any edge $(a,b)$ on $\delta_{uv}$ that is not in $G$, there exists a path in $G$ whose length is at most 3 times $\dist{a}{b}$. Thus, $\delta_{uv}$ can be turned into a path $\delta'_{uv}$ in $G$ such that $|\delta'_{uv}|\leqslant 3\sqrt{2}|uv|$; see Figure~\ref{lattice-stretch-fig}. Therefore, the stretch factor of $G$ is at most $3\sqrt{2}$.
\end{proof}
\section{Final Remarks}
\label{Steiner-section}
In order to obtain plane spanners with small stretch factor, one may think of adding Steiner points\footnote{some points in the plane that do not belong to the input point set.} to the point set and build a spanner on the augmented point set.
In the $L_1$-metric, a plane 1-spanner of degree 4 can be computed by using $O(n\log n)$ Steiner points~(see~\cite{Gudmundsson2007}). Arikati~{\em et al.}~\cite{Arikati1996} showed how to compute, in the $L_1$-metric,  a plane $(1+\epsilon)$-spanner with $O(n)$ Steiner points, for any $\epsilon>0$. Moreover, for the Euclidean metric, they showed how to construct a plane $(\sqrt{2}+\epsilon)$-spanner that uses $O(n)$ Steiner points and has degree 4. 

Let $S$ be a set of $n$ points in the plane that is in general position; no three points are collinear. Let $G$ be a plane $t$-spanner of $S$. We show how to construct, from $G$, a plane $(t+\epsilon)$-spanner of degree 3 for $S$ with $O(n)$ Steiner points, for any $\epsilon>0$. Let $CP$ denotes the closest pair distance in $S$ and let $\epsilon'=\epsilon\cdot CP$. 

\begin{wrapfigure}{r}{.31\textwidth} 
\vspace{-13pt} 
\centering
\includegraphics[draft=false, width=.29\textwidth]{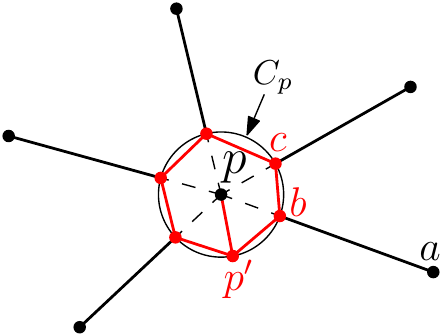} 
\vspace{-5pt} 
\end{wrapfigure}
For each point $p$ of the point set $S$, consider a circle $C_p$ with radius $\frac{\epsilon'}{\pi n}$ that is centered at $p$. Introduce a Steiner point on each intersection point of $C_p$ with the edges of $G$ that are incident on $p$. Also, introduce a Steiner point $p'$ on $C_p$ that is different from these intersection points. Delete the part of the edges of $G$ inside each circle $C_p$ (each edge $e=(p,q)$ of $G$ turns into an edge $e'$ of $G'$ with endpoints on $C_p$ and $C_q$).  Add an edge from $p$ to $p'$, and add a cycle whose edges connect consecutive Steiner points on the boundary of $C_p$. This results in a degree-3 geometric plane graph $G'$.
For each vertex of degree $k$ in $G$, we added $k+1$ Steiner points in $G'$. Since $G$ is planar, its total vertex degree is at most $6n-12$. Thus, the number of Steiner points is $7n-12$, in total (by removing Steiner vertices $p'$ and $b$ together with the edges incident on them in the above figure, keeping the edge $(a,p)$ of $G$, and adding the edge $(p,c)$ we get a different construction of $G'$ that uses $5n-12$ Steiner vertices).

A path $\delta_{uv}$ between two vertices $u$ and $v$ in $G$ can be turned into a path $\delta'_{uv}$ in $G'$ as follows. For each point $p$ in $S$ corresponding to an internal vertex of $\delta_{uv}$ incident on two edges $e_1$ and $e_2$ of $\delta_{uv}$, replace the part of $e_1$ and $e_2$ inside $C_p$ by the shorter of the two paths along $C_p$ connecting the corresponding Steiner points. Also, for each point $p\in\{u,v\}$ incident on an edge $e$ of $\delta_{uv}$ replace the part of $e$ inside $C_p$ with edge $(p,p')$ together with the shorter of the two paths along $C_p$ connecting $p'$ and the Steiner point corresponding to $e$.

Since $G$ is a $t$-spanner, $\frac{|\delta_{uv}|}{|uv|}\leqslant t$. Since the Steiner points are located at distance $\frac{\epsilon'}{\pi n}$ from points of $S$, the length of the detour caused by $C_p$ at each vertex is at most $\frac{\epsilon'}{n}$. Since $\delta_{uv}$ has at most $n$ vertices, the length of $\delta'_{uv}$ in $G'$ is at most $|\delta_{uv}|+n\cdot\frac{\epsilon'}{n}$. Thus,
$\frac{|\delta'_{uv}|}{|uv|}\leqslant\frac{|\delta_{uv}|+\epsilon'}{|uv|}=\frac{|\delta_{uv}|}{|uv|}+\epsilon\cdot\frac{CP}{|uv|}\leqslant t+\epsilon,$
is valid because the closet pair distance $CP$ is smaller than $|uv|$. 
\section*{Acknowledgement}
Some parts of this work have been done at the {\em Fourth Annual Workshop on Geometry and Graphs}, held at the Bellairs Research Institute in Barbados, March 6-11, 2016. The authors are grateful to the organizers and to the participants of this workshop. Also, we would like to thank G\"{u}nter Rote for his comments that simplified the proofs of Theorems~\ref{diametral-lune-thr} and \ref{diametral-circle-thr}, and an anonymous referee for simplifying the construction of the spanner for non-uniform grids.

\bibliographystyle{abbrv}
\bibliography{Deg-3-Spanner.bib}
\newpage
\appendix
\section{Upper bound for Function~(3)}
\label{appA}

We prove that
\begin{align}
\notag
f(x,\alpha)=\frac{1-x+\alpha}{\sqrt{(x\sin\alpha)^2+(1-x\cos\alpha)^2}} \leqslant \sqrt{1+\frac{1}{27}\left(3+2\pi\right)^2} \approx 2.04738
\end{align}
for all $0 < x \leqslant 1$, $0< \alpha\leqslant\frac{\pi}{3}$.
We rewrite~$f(x,\alpha)$ as
\begin{align}
\notag\frac{1-x+\alpha}{\sqrt{(x\sin\alpha)^2+(1-x\cos\alpha)^2}}
&\notag= \frac{1-x\sqrt{\sin^2\!\alpha+\cos^2\!\alpha}+\alpha}{\sqrt{(x\sin\alpha)^2+(1-x\cos\alpha)^2}} \\
&\notag= \frac{1-\sqrt{(x\sin\alpha)^2+(x\cos\alpha)^2}+\alpha}{\sqrt{(x\sin\alpha)^2+(1-x\cos\alpha)^2}} \\
&\notag= \frac{1-\sqrt{a^2+b^2}+\alpha}{\sqrt{a^2+(1-b)^2}} ,
\end{align}
where $a = x\sin\alpha$ and $b=x\cos\alpha$.
Note that $a^2+(1-b)^2>0$, because $\alpha>0$. We prove the following equivalent inequality:
\begin{align}
\notag
f(\alpha,a,b) = \frac{1-\sqrt{a^2+b^2}+\alpha}{\sqrt{a^2+(1-b)^2}} \leqslant \sqrt{1+\frac{1}{27}\left(3+2\pi\right)^2} \approx 2.04738
\end{align}
for all $a = b\tan\alpha$,
$0< b \leqslant \cos\alpha$,
$0< \alpha\leqslant\frac{\pi}{3}$.
We have
$$\frac{\partial}{\partial \alpha} f(\alpha,a,b) = \frac{1}{\sqrt{a^2+(1-b)^2}} .$$
Therefore,
the system
$$\frac{\partial}{\partial \alpha} f(\alpha,a,b) = \frac{\partial}{\partial a} f(\alpha,a,b) = \frac{\partial}{\partial b} f(\alpha,a,b) = 0$$
does not have a solution.
Thus,
we look at the boundary conditions.
Since $a = b\tan\alpha$,
the boundary conditions are (1) $\alpha = 0$, (2) $\alpha = \frac{\pi}{3}$, (3) $b=0$, and (4) $b = \cos\alpha$.
\begin{enumerate}
\item Since $a = b\tan\alpha$,
we have
$$f(0,a,b) = f(0,0,b) = 1 < \sqrt{1+\frac{1}{27}\left(3+2\pi\right)^2} .$$

\item Since $a = b\tan\alpha$
and using elementary calculus,
we can show that
\begin{align*}
f\left(\frac{\pi}{3},a,b\right) &= f\left(\frac{\pi}{3},\sqrt{3}\,b,b\right) = \frac{1 -2b +\frac{\pi}{3}}{\sqrt{3 b^2 + (1 - b)^2}} \\
&\leqslant f\left(\frac{\pi}{3},\sqrt{3}\,\frac{\pi-3}{4\pi+6},\frac{\pi-3}{4\pi+6}\right) = \sqrt{1+\frac{1}{27}\left(3+2\pi\right)^2}
\end{align*}
for all $0< b \leqslant \cos\left(\frac{\pi}{3}\right)$.

\item Since $a = b\tan\alpha$,
we have
$$f(\alpha,a,0) = f(\alpha,0,0) = 1 + \alpha \leqslant 1+\frac{\pi}{3} < \sqrt{1+\frac{1}{27}\left(3+2\pi\right)^2} $$
for all $0< \alpha \leqslant \frac{\pi}{3}$.

\item Since $a = b\tan\alpha$
and using elementary calculus,
we can show that
\begin{align*}
f\left(\alpha,a,\cos\alpha\right) &= f\left(\alpha,\sin\alpha,\cos\alpha\right) = \frac{\alpha}{\sqrt{2-2\cos\alpha}} \\
&\leqslant f\left(\frac{\pi}{3},\frac{\sqrt{3}}{2},\frac{1}{2}\right) = \frac{\pi}{3} < \sqrt{1+\frac{1}{27}\left(3+2\pi\right)^2}
\end{align*}
for all $0< \alpha \leqslant \frac{\pi}{3}$.
\end{enumerate}

\section{Upper bound for Function~(4)}
\label{appB}

We prove that
\begin{align}
\notag
g(x,\alpha)=\frac{2(\alpha + \cos\alpha)-\left(x+\frac{\pi}{3}\right)}{\sqrt{(x\sin\alpha)^2+(1-x\cos\alpha)^2}} \leqslant \frac{2\pi}{3} \approx 2.09440
\end{align}
for all $0 < x \leqslant 2\cos\alpha$, $\frac{\pi}{3}\leqslant \alpha\leqslant\frac{\pi}{2}$.
We rewrite~$g(x,\alpha)$ as
\begin{align*}
\frac{2(\alpha + \cos\alpha)-\left(x+\frac{\pi}{3}\right)}{\sqrt{(x\sin\alpha)^2+(1-x\cos\alpha)^2}} &=
\frac{2\left(\alpha + \frac{1}{\sqrt{\tan^2\!\alpha+1}}\right)-\left(x\sqrt{\sin^2\!\alpha+\cos^2\!\alpha}+\frac{\pi}{3}\right)}{\sqrt{(x\sin\alpha)^2+(1-x\cos\alpha)^2}} \\
&=
\frac{2\left(\alpha + \frac{1}{\sqrt{\left(\frac{x\sin\alpha}{x\cos\alpha}\right)^2+1}}\right)-\left(\sqrt{(x\sin\alpha)^2+(x\cos\alpha)^2}+\frac{\pi}{3}\right)}{\sqrt{(x\sin\alpha)^2+(1-x\cos\alpha)^2}} \\
&=
\frac{2\left(\alpha + \frac{1}{\sqrt{\left(\frac{a}{b}\right)^2+1}}\right)-\left(\sqrt{a^2+b^2}+\frac{\pi}{3}\right)}{\sqrt{a^2+(1-b)^2}} \\
&=
\frac{2\left(\alpha + \frac{b}{\sqrt{a^2+b^2}}\right)-\left(\sqrt{a^2+b^2}+\frac{\pi}{3}\right)}{\sqrt{a^2+(1-b)^2}} \\
&=
\frac{6\left(\alpha + \frac{b}{\sqrt{a^2+b^2}}\right)-\left(3\sqrt{a^2+b^2}+\pi\right)}{3\sqrt{a^2+(1-b)^2}} \\
&=
\frac{6\alpha-\pi}{3\sqrt{a^2+(1-b)^2}}-\frac{3\sqrt{a^2+b^2}-\frac{6b}{\sqrt{a^2+b^2}}}{{3\sqrt{a^2+(1-b)^2}}} \\
&=
\frac{6\alpha-\pi}{3\sqrt{a^2+(1-b)^2}}-\frac{a^2 +(1-b)^2-1}{\sqrt{a^2+b^2}\sqrt{a^2+(1-b)^2}} ,
\end{align*}
where $a = x\sin\alpha$ and $b=x\cos\alpha$. Note that $a^2+b^2>0$ and $a^2+(1-b)^2>0$, because $x>0$ and $\alpha>0$.
We prove the following equivalent inequality:
\begin{align}
\notag
g(\alpha,a,b) = \frac{6\alpha-\pi}{3\sqrt{a^2+(1-b)^2}}-\frac{a^2 +(1-b)^2-1}{\sqrt{a^2+b^2}\sqrt{a^2+(1-b)^2}} \leqslant \frac{2\pi}{3} \approx 2.09440
\end{align}
for all $a = b\tan\alpha$,
$0\leqslant b \leqslant 2\cos^2\!\alpha$,
$\frac{\pi}{3} \leqslant \alpha \leqslant\frac{\pi}{2}$.
We have
$$\frac{\partial}{\partial \alpha} g(\alpha,a,b) = \frac{2}{\sqrt{a^2+(1-b)^2}} .$$
Therefore,
the system
$$\frac{\partial}{\partial \alpha} g(\alpha,a,b) = \frac{\partial}{\partial a} g(\alpha,a,b) = \frac{\partial}{\partial b} g(\alpha,a,b) = 0$$
does not have a solution.
Thus,
we look at the boundary conditions.
Since $a = b\tan\alpha$,
the boundary conditions are (1) $\alpha = \frac{\pi}{3}$, (2) $\alpha = \frac{\pi}{2}$, (3) $b=0$, and (4) $b = 2\cos^2\!\alpha$.
\begin{enumerate}
\item Since $a = b\tan\alpha$
and using elementary calculus,
we can show that
\begin{align*}
g\left(\frac{\pi}{3},a,b\right) &= g\left(\frac{\pi}{3},\sqrt{3}\,b,b\right) = \frac{1 -2b +\frac{\pi}{3}}{\sqrt{3 b^2 + (1 - b)^2}} \\
&\leqslant g\left(\frac{\pi}{3},\sqrt{3}\,\frac{\pi-3}{4\pi+6},\frac{\pi-3}{4\pi+6}\right) = \sqrt{1+\frac{1}{27}\left(3+2\pi\right)^2} < \frac{2\pi}{3}
\end{align*}
for all $0\leqslant b \leqslant 2\cos^2\!\left(\frac{\pi}{3}\right)$.

\item When $\alpha = \frac{\pi}{2}$,
$a = b\tan\alpha$ is not well-defined.
Instead,
we write $b = a\,\cot\alpha = 0$
and $0 \leqslant a \leqslant 2\sin\alpha\cos\alpha = 0$,
from which $a = b = 0$.
We have
$$g\left(\frac{\pi}{2},0,0\right) = \frac{2\pi}{3} .$$

\item Since $a = b\tan\alpha$,
we have
\begin{align*}
g(\alpha,a,0) &= g(\alpha,0,0) = \frac{6\alpha-\pi}{3} \leqslant \frac{2\pi}{3}
\end{align*}
for all $\frac{\pi}{3} \leqslant \alpha \leqslant \frac{\pi}{2}$.

\item Since $a = b\tan\alpha$,
we have
\begin{align*}
g\left(\alpha,a,2\cos^2\!\alpha\right) &= g\left(\alpha,2\sin\alpha\cos\alpha,2\cos^2\!\alpha\right) = 2\alpha-\frac{\pi}{3} \leqslant \frac{2\pi}{3}
\end{align*}
for all $\frac{\pi}{3} \leqslant \alpha \leqslant \frac{\pi}{2}$.
\end{enumerate}
\end{document}